\documentclass[runningheads]{llncs}

\usepackage[T1]{fontenc}
\usepackage{graphicx}
\usepackage[misc]{ifsym}

\usepackage{bbding}

\usepackage{amsmath}
\usepackage{booktabs}   
\usepackage{subcaption}
\usepackage{enumerate}
\usepackage{tikz}
\usetikzlibrary{backgrounds, bending}
\usetikzlibrary{shapes.geometric}
\usepackage{listings}
\usepackage{csvsimple} 
\usepackage{multirow}
\usepackage{amsfonts}
\usepackage{url}
\usepackage{babel}
\usepackage{mathrsfs}
\usepackage[strings]{underscore}
\usepackage{enumerate,paralist}
\usepackage{lstautogobble}
\usepackage{makecell}
\usepackage{pifont}
\usepackage{paralist}
\usepackage{hyperref}
\usepackage{microtype}
\usepackage{float}
\usepackage{enumitem}
\usepackage{hhline}
\usepackage{tabularx}

\usepackage[symbol]{footmisc}

\renewcommand{\paragraph}[1]{\smallskip\noindent\textbf{\emph{#1}.}}

\newcommand{\R}{\mathbb{R}}
\newcommand{\N}{\mathbb{N}}
\newcommand{\V}{\mathcal{V}}
\renewcommand{\S}{\mathscr{S}}

\newcommand{\B}{\mathcal{B}}
\newcommand{\T}{\mathcal{T}}
\newcommand{\term}{\mathit{term}}
\newcommand{\init}{\mathit{init}}

\newcommand{\dist}{\mathit{dist}}
\newcommand{\at}{\mathit{at}}
\newcommand{\Out}{\mathit{Out}}

\newcommand{\xmark}{\ding{55}}%
\newcommand{\AP}{\texttt{AP}}
\newcommand{\F}{\texttt{F}}
\newcommand{\G}{\texttt{G}}
\newcommand{\U}{\texttt{U}}
\newcommand{\X}{\texttt{X}}
\newcommand{\formula}{\varphi}
\newcommand{\run}{\pi}
\newcommand{\scheduler}{\sigma}
\newcommand{\states}{\S}
\newcommand{\variables}{\V}
\newcommand{\locations}{L}
\newcommand{\initloc}{l_\init}
\newcommand{\precondition}{\theta_\init}

\newcommand{\guard}{G}
\newcommand{\update}{U}
\newcommand{\initstates}{\states_\init}
\newcommand{\terminal}{l_t}
\newcommand{\powerset}[1]{2^{#1}}
\newcommand{\nbastates}{Q}
\newcommand{\nbaalpha}{A}
\newcommand{\nbatrans}{\delta}
\newcommand{\nbainit}{q_0}
\newcommand{\nbaaccept}{F}
\newcommand{\invariant}{\theta}
\newcommand{\outtrans}{\Out}
\newcommand{\tstrans}{\mapsto}

\definecolor{codegreen}{rgb}{0,0.6,0}
\definecolor{codegray}{rgb}{0.5,0.5,0.5}
\definecolor{codepurple}{rgb}{0.58,0,0.82}

\lstdefinestyle{mystyle}
{
	language = python,
	basicstyle = {\ttfamily\footnotesize},
	stringstyle = {\color{string-color}},
	keywordstyle = {\bfseries\color{codegreen}},
	keywordstyle = [2]{\bfseries\color{codegreen}},
	keywordstyle = [3]{\bfseries\color{codepurple}},
	keywordstyle = [4]{\bfseries\color{teal}},
	otherkeywords = {;,<<,>>,++},
	morekeywords = [2]{do, then, fi, done},
	morekeywords = [3]{true, false, skip},
	morekeywords = [4]{Precondition},
}

\AtBeginDocument{%
	}

\begin{document}
	
	\title{Sound and Complete Witnesses for Template-based Verification of LTL Properties on Polynomial Programs}
	
\author{Krishnendu Chatterjee\inst{1}\orcidID{0000-0002-4561-241X} 
\and
Amir Kafshdar Goharshady\inst{2}\orcidID{0000-0002-2643-120X} 
\and
Ehsan Kafshdar Goharshady\inst{1}\orcidID{0000-0002-8595-0587} 
\and
Mehrdad Karrabi(\Letter)\inst{1}\orcidID{0009-0007-5253-9170} 
\and
\DJ{}or\dj{}e \v{Z}ikeli\'c\inst{3}\orcidID{0000-0002-4681-1699}\thanks{Part of the work done while the author was at the Institute of Science and Technology Austria (ISTA).}
}
\authorrunning{K. Chatterjee et al.}

\titlerunning{Sound and Complete Witnesses for LTL Properties on Polynomial Programs}

\institute{Institute of Science and Technology Austria (ISTA), Klosterneuburg, Austria \email{\{krishnendu.chatterjee, ehsan.goharshady, mehrdad.karrabi\}@ist.ac.at} \and
The Hong Kong University of Science and Technology (HKUST), Hong Kong\\ \email{goharshady@cse.ust.hk}
\and
Singapore Management University, Singapore\\
\email{dzikelic@smu.edu.sg}}

\authorrunning{K.~Chatterjee, A.~K.~Goharshady, E.~K.~Goharshady, M.~Karrabi, \DJ.~\v{Z}ikeli\'c}

		\maketitle
	\begin{abstract}
		We study the classical problem of verifying programs with
		respect to formal specifications given in the linear temporal logic (LTL).
		We first present novel sound and complete witnesses for LTL verification
		over imperative programs. Our witnesses are applicable to both verification
		(proving) and refutation (finding bugs) settings. We then consider LTL
		formulas in which atomic propositions can be polynomial constraints and
		turn our focus to polynomial arithmetic programs, i.e. programs in which
		every assignment and guard consists only of polynomial expressions. For
		this setting, we provide an efficient algorithm to automatically synthesize
		such LTL witnesses. Our synthesis procedure is both sound and semi-complete. 
        Finally, we present experimental results demonstrating the
		effectiveness of our approach and that it can handle programs which were
		beyond the reach of previous state-of-the-art tools.
	\end{abstract}



\section{Introduction} 

\paragraph{Linear-time Temporal Logic} The Linear-time Temporal Logic (LTL)~\cite{ltl} is one of the most classical and well-studied frameworks for formal specification, model checking and program verification. In LTL, we consider a set $\AP$ of atomic propositions and an infinite trace which tells us which propositions in $\AP$ hold at any given time. LTL formulas are then able to not only express propositional logical operations, but also modalities referring to the future. For example, $\X\ p$ requires that $p$ holds in the next timeslot, whereas $\F\ q$ means $q$ should hold at some time in the future. This allows LTL to express common verification tasks such as termination, liveness, fairness and safety. 


\paragraph{Witnesses} Given a specification $\formula$ and a program $P,$ a \emph{witness} is a mathematical object whose existence proves that the specification $\formula$ is satisfied by $P$. We say that a witness family is \emph{sound and complete} when for every program $P$ and specification $\formula,$ we have $P \models \formula$ if and only if there is a witness in the family that certifies it. 
Witnesses are especially useful in dealing with undecidable problems in verification, which includes all non-trivial semantic properties~\cite{rice1953classes}. This is because although the general case of the problem is undecidable, having a sound and complete notion of a witness can lead to algorithms that check for the existence of witnesses of a special form. For example, while termination is undecidable~\cite{turing1936computable}, and hence so is the equivalent problem of deciding the existence of a ranking function, there are nevertheless sound and complete algorithms for synthesis of \emph{linear} ranking functions~\cite{podelski2004complete}. Similarly, while reachability (safety violation) is undecidable, it has sound and complete witnesses that can be automatically synthesized in linear and polynomial forms~\cite{reach20}. Our work subsumes both~\cite{podelski2004complete} and~\cite{reach20} and provides sound and complete witnesses for general LTL formulas.

\paragraph{Polynomial Programs} In this work, we mainly focus on imperative programs with polynomial arithmetic. More specifically, our programs have real variables and the right-hand-side of every assignment is a polynomial expression with respect to program variables. Similarly, the guard of every loop or branch is also a boolean combination of polynomial inequalities over the program variables. 

\paragraph{Our Contributions} In this work, our contributions are threefold: 
\begin{compactitem}
	\item On the theoretical side, by exploiting the connections to B\"uchi automata, we propose a novel family of sound and complete witnesses for general LTL formulas. This extends and unifies the known concepts of ranking functions~\cite{floyd1993assigning}, inductive reachability witnesses~\cite{reach20} and inductive invariants~\cite{colon2003linear}, which are sound and complete witnesses for termination, reachability and safety, respectively. Our theoretical result is not limited to polynomial programs. 
	
	\item On the algorithmic side, we consider polynomial programs and present a sound and semi-complete template-based algorithm to synthesize polynomial LTL witnesses. This algorithm is a generalization of the template-based approaches in~\cite{podelski2004complete,reach20,colon2003linear} which considered termination, reachability and safety. To the best of our knowledge, this is the most general model checking problem over polynomial programs to be handled by template-based approaches~to~date.
	
	\item Finally, on the experimental side, we provide an implementation of our approach and comparisons with state-of-the-art LTL model checking tools. Our experiments show that our approach is applicable in practice and can handle many instances that were beyond the reach of previous methods. Thus, our completeness result pays off in practice and enables us to solve new instances.
\end{compactitem}

\paragraph{Motivation for Polynomial Programs} There are several reasons why we consider polynomial programs:
\begin{compactitem}
	\item Many real-world families of programs, such as, programs for cyber-physical systems and smart contracts, can be modeled in this framework~\cite{DBLP:conf/iccps/GurrietSRCFA18,fulton2018verifiably,oopsla}.
	\item They are one of the most general families for which finding polynomial witnesses for reachability and safety are known to be decidable~\cite{DBLP:conf/popl/SankaranarayananSM04,reach20,invariant20}. Hence, they provide a desirable tradeoff between decidability and generality.
	\item  Using abstract interpretation, non-polynomial behavior in a program can be removed or replaced by non-determinism. Moreover, one can approximate any continuous function up to any desired level of accuracy by a polynomial. This is due to the Stone--Weierstrass theorem~\cite{de1959stone}. Thus, analysis of polynomial programs can potentially be applied to many non-polynomial programs via abstract interpretation or numerical approximation of the program's behavior. 
	\item Previous works have studied (a)~linear/affine programs with termination, safety, and reachability specifications~\cite{podelski2004complete,colon2003linear,DBLP:conf/sas/SankaranarayananSM04}, and (b)~polynomial programs with termination, safety and reachability properties~\cite{DBLP:conf/popl/SankaranarayananSM04,invariant20,reach20,DBLP:conf/cav/ChatterjeeFG16}. Since LTL subsumes all these specifications, polynomial program analysis with LTL provides a unifying and general framework for all these previous works. 
\end{compactitem}

\paragraph{Related Works on Linear Programs} There are many approaches focusing on linear witness synthesis for important special cases of LTL formulas. For example,~\cite{podelski2004complete,DBLP:conf/atva/HeizmannHLP13} consider the problem of synthesizing linear ranking functions (termination witnesses) over linear arithmetic programs. The works~\cite{colon2003linear,DBLP:conf/sas/SankaranarayananSM04} synthesize linear inductive invariants (safety witnesses), while~\cite{DBLP:conf/tacas/FunkeJB20} considers probabilistic reachability witnesses.  The work~\cite{DBLP:conf/pldi/GulwaniSV08} handles a larger set of verification tasks and richer settings, such as context-sensitive interprocedural program analysis. All these works rely on the well-known Farkas lemma~\cite{farkas} and can handle programs with linear/affine arithmetic and synthesize linear/affine witnesses. In comparison, our approach is (i)~applicable to general LTL formulas and not limited to a specific formula such as termination or safety, and (ii)~able to synthesize \emph{polynomial} witnesses for \emph{polynomial} programs with soundness and completeness guarantees. Thus, our setting is more general in terms of (a)~formulas, (b)~witnesses, and (c)~programs that can be supported.

\paragraph{Related Works on Polynomial Programs} Similar to the linear case, there is a rich literature on synthesis of polynomial witnesses over polynomial programs. However, these works again focus on specific special formulas only and are not applicable to general LTL. For example,~\cite{DBLP:conf/concur/NeumannO020,DBLP:conf/fm/MoosbruggerBKK21,DBLP:conf/cav/ChatterjeeFG16,DBLP:journals/jossac/ShenWYZ13,nonterm21} consider termination analysis,~\cite{invariant20} extends the invariant generation (safety witness synthesis) algorithm of~\cite{colon2003linear} to the polynomial case and~\cite{DBLP:conf/atva/FengZJZX17} further adds support for probabilistic programs. The works~\cite{DBLP:conf/cdc/Clark21,DBLP:journals/tcad/ZhangYLZCL18,DBLP:conf/cav/WangCXZK21} consider alternative types of witnesses for safety (barriers) and obtain similarly successful synthesis algorithms. Finally,~\cite{DBLP:journals/toplas/TakisakaOUH21,reach20} synthesize reachability witnesses. Since we can handle any arbitrary LTL formula, our approach can be seen as an extension and unification of all these works. Indeed, our synthesis algorithm directly builds upon and extends~\cite{reach20}.

In both cases above, some of the previous works are incomparable to ours since they consider probabilistic programs, whereas our setting has only non-probabilistic polynomial programs. Note that we do allow non-determinism.

\paragraph{Related Works on LTL Model Checking} There are thousands of works on LTL model checking and there is no way we can do justice to all. We refer to~\cite{henzinger2018handbook,strejcek2004linear} for an excellent treatment of the finite-state cases. Some works that provide LTL model checking over infinite-state systems/programs are as follows:
\begin{compactitem}
	\item A prominent technique in this area is predicate abstraction~\cite{DBLP:conf/cav/GrafS97,DBLP:conf/cav/DanielCGTM16,PodelskiR05}, which uses a finite set of abstract states defined by an equivalence relation based on a finite set of predicates to soundly, but not completely, reduce the problem to the finite-state case. 
	\item \cite{F3paper} uses a compositional approach to falsify LTL formulas and find an indirect description of a path that violates the specification.
	\item There are several symbolic approaches, including~\cite{DBLP:conf/tacas/CookKP15} which is focused on fairness and~\cite{DBLP:conf/memics/BauchHB14} which is applicable to LLVM. Another work in this category is \cite{UltimatePaper}, whose approach is to repeatedly rule out infeasible finite prefixes in order to find a run of the program that satisfies/violates the desired LTL formula.
	The work \cite{DBLP:conf/popl/CookK11} uses CTL-based approaches that might report false counter-examples when applied to LTL. It then identifies and removes such spurious counterexamples using symbolic determinization.
	
	\item The work~\cite{FarzanKP16} presents a framework for proving liveness properties in multi-threaded programs by using well-founded proof spaces.
	
	\item The recent work~\cite{DBLP:journals/fmsd/PadonHMPSS21} uses temporal prophecies, inspired by classical prophecy variables, to provide significantly more precise reductions from general temporal verification to the special case of safety. 
	
	\item There are many tools for LTL-based program analysis. For example, T2~\cite{t2} is able to verify a large family of liveness and safety properties, nuXmv~\cite{nuxmv} is a symbolic model checker with support for LTL, F3~\cite{F3paper} proves fairness in infinite-state transition systems, and Ultimate LTLAutomizer~\cite{UltimatePaper} is a general-purpose tool for verification of LTL specifications over a wide family of programs with support for various types of variables. 
	
	\item Finally, we compare against the most recent related work~\cite{unno2023muval}. This work provides relative-completeness guarantees for general programs with LTL specifications. Since it considers integer programs with recursive functions, there is no complexity guarantee provided. The earlier work~\cite{unno2021csp} provides several special cases where termination is guaranteed. However, no runtime bounds are established. In contrast, our approach has both termination guarantees and sub-exponential time complexity for fixed degree. 
\end{compactitem}
As shown by our experimental results in Section~\ref{sec:exper}, our completeness results enable our tool to handle instances that other approaches could not. On the other hand, our method is limited to polynomial programs and witnesses. Thus, there are also cases in which our approach fails but some of the previous tools succeed, e.g.~when the underlying program requires a non-polynomial witness. In particular, Ultimate LTLAutomizer~\cite{UltimatePaper} is able to handle non-polynomial programs and witnesses, too.

	\section{Transition Systems, LTL and B\"uchi Automata}
\label{prelim}

For a vector $e \in \R^n$ , we use $e_i$ to denote the $i$-th component of $e$. Given a finite set $\V$ of real-valued variables, a variable valuation $e\in\R^{|\V|}$ and a boolean predicate $\formula$ over $\V$, we write $e\models \formula$ when $\formula$ evaluates to true upon substituting variables by the values given in $e$.


We consider imperative numerical programs with real-valued variables, containing standard programming constructs such as assignments, branching and loops. In addition, our programs can have finite non-determinism. We denote non-deterministic branching in our syntax by  \textbf{if $\ast$ then}. See Figure~\ref{example_ts} for an example. We use transition systems to formally model programs.

\paragraph{Transition systems} 
An infinite-state \emph{transition system} is a tuple $\T = (\variables, \locations,  \initloc, \precondition, \tstrans)$, where:
	\begin{compactitem}
	\item $\variables = \{x_0, \dots, x_{n-1}\}$ is a finite set of real-valued {\em program variables}.
	\item $\locations$ is a finite set of {\em locations} with $\initloc \in L$ the {\em initial location}.
	\item $\precondition \subseteq \mathbb{R}^n$ is a set of {\em initial variable valuations}.
	\item $\tstrans$ is a finite set of {\em transitions}. Each transition $\tau \in \,\tstrans$ is of the form $\tau = (l, l', \guard_\tau, \update_\tau)$, where $l$ is the source location, $l'$ is the target location, $\guard_\tau$ is the guard of the transition, which is a boolean predicate over $\variables,$ and $\update_\tau: \R^n \to \R^n$ is the update function of the transition. 
	\end{compactitem}
Translating programs into transition systems is a standard process. In what follows, we assume we are given a transition system $\T = (\variables, \locations,  \initloc, \precondition, \tstrans)$ of the program that we wish to analyze. An example is shown in Figure~\ref{example_ts}.

\paragraph{States and Runs} A {\em state} in $\T$ is a pair $(l,e)$ with $l\in \locations$ and $e\in\R^n$. A state $(l,e)$ is said to be {\em initial} if $l=\initloc$ and $e\in\precondition$. We use $\states$ and $\initstates$ to denote the sets of all states and initial states. We assume the existence of a special {\em terminal location} $\terminal$ with a single outgoing transition which is a self-loop $(\terminal,\terminal,\textrm{true},\mathit{Id})$ with $\mathit{Id}(e) = e$ for each $e\in\R^n$.
A state $(l',e')$ is a {\em successor} of $(l,e),$ denoted as $(l,e) \tstrans (l',e'),$ if there exists a transition $\tau = (l, l', \guard_\tau, \update_\tau) \in \tstrans$ such that $e \models \guard_\tau$ and $e' =  \update_\tau(e)$. We assume each state has at least one successor so that all runs are infinite and LTL semantics are well defined. This is without loss of generality, since we can introduce transitions to the terminal location. 
A {\em run} in $\T$ is an infinite sequence of successor states starting in $\initstates$. 

\begin{figure}[t]
	\centering
	\begin{subfigure}{0.3\textwidth}
		\begin{lstlisting}[frame=none,numbers=none]
Precondition $(l_{init})$: $x_0 \geq 0$
$l_1$:  while $x_0 \geq 0$ do
     if $\ast$ then
$l_2$:      $x_0 = x_0^2 + 1$
     else 
$l_3$:      $x_0 = x_0^2 - 1$
$l_t$:
		\end{lstlisting}
	\end{subfigure}%
\hfill
	\begin{subfigure}{0.55\textwidth}
%
%
%
\includegraphics[scale=0.60]{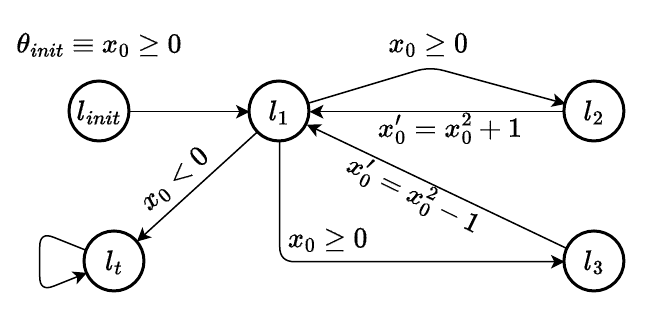}
	\end{subfigure}
	\caption{An example program (left) and its transition system (right). Note that there is non-determinism at $l_1.$}
	\label{example_ts}
\end{figure}

\paragraph{Linear-time Temporal Logic (LTL)}
	Let $\AP$ be a finite set of atomic propositions. LTL formulas are inductively defined as follows:
	\begin{compactitem}
		\item If $p \in \AP$, then $p$ is an LTL formula.
		\item If $\formula$ and $\psi$ are LTL formulas, then $\neg \formula$, $\formula \vee \psi$,  $\formula \wedge \psi$, $\X\ \formula$, $\G\ \formula$, $\F\ \formula$  and $\formula\ \U\ \psi$ are all LTL formulas.
	\end{compactitem}
	$\neg, \vee$ and $\wedge$ are the propositional negation, disjunction and conjunction while $\X,\G,\F$ and $\U$ are the \textit{next, globally, finally} and \textit{until} temporal operators. 

\paragraph{Atomic Propositions} To use LTL over the transition system $\T$, we first need to specify a finite set of atomic propositions $\AP$. In this work, we let the set $\AP$ consist of (i)~finitely many constraints of the form $\textrm{exp}(\mathbf{x})\geq 0$ where $\textrm{exp}\colon\V\rightarrow\R$ is an arithmetic expression over $\V$, and (ii)~an atomic proposition $\at(l)$ for each location $l$ in $\T$. 
Note that unlike classical LTL settings, our atomic propositions are not necessarily independent. For example, if we have $p_1 := x \geq 0$ and $p_2 := x+1 \geq 0,$ it is impossible to have $p_1 \wedge \neg p_2$ at any point in time.

The semantics of LTL is standard, refer to Appendix \ref{app:sem} for details.

\paragraph{Program Analysis with LTL Specifications}
We now define the LTL program analysis problems that we consider in this work. Given a transition system $\T$ and an LTL formula $\formula$, we are interested in two problems:
\begin{compactenum}
	\item {\em LTL Verification of Programs (LTL-VP).} Given a transition system $\T$ and an LTL formula $\formula$ in $\T$, prove that \emph{all possible runs} of $\T$ satisfy $\formula$.
	\item {\em LTL Refutation of Programs (LTL-RP).} Given a transition system $\T$ and an LTL formula $\formula$ in $\T$, prove that there \emph{exists a run} that violates $\formula$, or equivalently, satisfies $\neg \formula$.
\end{compactenum}

\paragraph{Remark} \textit{LTL Verification} asks about correctness of the program while \textit{LTL Refutation} addresses the problem of finding bugs. Both problems have been widely studied in the literature \cite{BaresiKR15,UltimatePaper,unno2023muval}. Moreover, a witness for the refutation problem can be used in counterexample-guided techniques such as CEGAR \cite{ClarkeGJLV00}.

\paragraph{Example}
	Consider the transition system in Figure \ref{example_ts} and the LTL formula $\formula =\neg [\G (\at(l_3) \Rightarrow \F \at(l_2))]$. The run that starts at $(\initloc, 1)$ and chooses $l_2$ if $x_0 = 0$ and $l_3$ whenever $x_0 = 1$, does not satisfy $\formula$. Therefore, in this case, the answer to the LTL-RP problem is positive. Additionally, deciding termination of a program with terminal location $\terminal$ is equivalent to the LTL-VP problem of $[\F\ \at(\terminal)]$ on the same program.

\begin{paragraph}{Program Analysis with Büchi Specifications}
A B\"uchi specification is a subset $\B \subseteq \states$ of states. A run $\run$ is {\em $\B-$Büchi} if it visits $\B$ infinitely many times, i.e.~if $\{ i ~|~ \run(i) \in \B \}$ is infinite. Similar to LTL, Büchi specifications give rise to two main decision problems as follows:
	\begin{compactenum}
		\item {\em Universal Büchi Program Analysis (UB-PA).} Given a transition system $\T$ and a Büchi specification $\B$ on $\T$, prove that \emph{all possible runs} of $\T$ are {\em $\B-$Büchi}.
		\item {\em Existential Büchi Program Analysis (EB-PA).} Given a transition system $\T$ and a Büchi specification $\B$ on $\T$, prove the \emph{existence} of a run that is {\em $\B-$Büchi}.
	\end{compactenum}
\end{paragraph} 


\paragraph{B\"uchi Automata~\cite{katoen08,buchi1966symposium}} A {\em non-deterministic Büchi automaton (NBW)} is a tuple $N=(\nbastates,\nbaalpha,\nbatrans, \nbainit, \nbaaccept)$, where 
$\nbastates$ is a finite set of states, $\nbaalpha$ is a finite alphabet, $\nbatrans\colon \nbastates \times \nbaalpha \to \powerset{\nbastates}$ is a  transition relation, $\nbainit$ is the initial state, and $\nbaaccept \subseteq \nbastates$ is the set of accepting states. 
An infinite word $a_0,a_1,\dots$ of letters in the alphabet $\nbaalpha$ is accepted by $N$ if it gives rise to at least one accepting run in $N$, i.e.~if there exists a run $q_0,q_1,\dots$ such that $q_{i+1} \in \nbatrans(q_i,a_i)$ for each $i$ and $\nbaaccept$ is visited infinitely many times. It is a classical result that for every LTL formula $\formula$ defined over atomic predicates $\AP$ there exists a non-deterministic Büchi automaton $N$ with alphabet $\powerset{AP}$ which accepts exactly those traces that satisfy $\formula$~\cite{henzinger2018handbook}.


Let $\T=(\variables, \locations,  \initloc, \precondition, \tstrans)$ be a transition system and $N=(\nbastates,\powerset{AP},\nbatrans, \nbainit, \nbaaccept)$ be an NBW. In order to analyse $\T$ with respect to $N$, we utilize the Cartesian product $\T \times N$ and the B\"uchi specification $\B^\T_N=\locations \times \nbaaccept \times \R^n$. The state space of $\T \times N$ is exactly the Cartesian product of the state spaces of $\T$ and $N$. Moreover, for $l,l' \in \locations$ and $q,q' \in \nbastates$, there is a transition from $(l,q)$ to $(l',q')$ if there is a transition in $\T$ from $l$ to $l'$ and a transition in $N$ from $q$ to $q'$. The formal definition of the product is available in Appendix \ref{app:LTL2BA}. See Figure~\ref{fig:2} for an example.

\begin{lemma}[From LTL to Büchi Specifications, Proof in Appendix~\ref{proof:lm1}]
	\label{reduction}
	Let $\T$ be a transition system, $\formula$ an LTL formula for $\T$ and $N$ an NBW that accepts the same language as $\formula$. 
 \begin{compactitem}
     \item The LTL-RP problem of $\T$ and $\neg \formula$ is equivalent to the EB-PA problem of $\T \times N$ and $\B^\T_N$~\cite{UltimatePaper}.
     \item If $N$ is deterministic, then the LTL-VP problem of $\T$ and $\formula$ is equivalent to the UB-PA problem of $\T \times N$ and $\B^\T_N$.
 \end{compactitem}
\end{lemma}


\input{example_mult}

\paragraph{Remark} Based on the lemma above, instead of designing witnesses for the LTL-RP problem, we only need to find sound and complete witnesses for EB-PA. Moreover, it is easy to see that LTL-VP is reducible to LTL-RP since all runs of $\T$ satisfy $\formula$ if and only if there is no run that satisfies $\neg \formula.$ So, finding sound and complete witnesses for EB-PA will theoretically solve both verification and refutation variants of LTL program analysis. 
Note that the second statement in Lemma~\ref{reduction} is more restrictive than the first one since it only applies to deterministic B\"uchi automata. Thus, if the LTL formula $\formula$ does not admit a deterministic B\"uchi automaton, the above sequence of reductions from LTL-VP to LTL-RP should be made and then the EB-PA witness should be used. However, if $\formula$ admits a DBW, then the reduction to UB-PA is preferable in practice. We will provide witness concepts for both EB-PA and UB-PA problems in the next section. 

	\section{Sound and Complete B-PA Witnesses}
\label{sec:witness}
Let $\T = (\variables, \locations, l_{\init}, \precondition, \tstrans)$ be a transition system and $\B \subseteq \states$ a set of states in $\T$. In this section, we introduce our sound and complete witnesses for the EB-PA and UB-PA problems.

\subsection{Sound and Complete Witnesses for Existential B-PA}
 Our witness concept for the EB-PA problem is a function that assigns a real value to each state in $\T$. The witness function is required to be non-negative in at least one initial state of $\T$, to preserve non-negativity in at least one successor state and to strictly decrease in value in at least one successor state whenever the current state is not contained in $\B$ and the value of the witness function in the current state is non-negative. Hence, starting in an initial state in which the witness function is non-negative, one can always select a successor state in which the witness function is non-negative and furthermore ensure that $\B$ is eventually reached due to the strict decrease condition, which will also be referred to as the \emph{B\"uchi-ranking condition}. Intuitively, an EBRF is a function that overestimates the distance to $\B$ and guarantees that $\B$ is reached along at least one program run, at every program state in which the value of the EBRF is non-negative. 

\begin{definition}[EBRF]
	\label{EBRF}
	Given two states $s_1, s_2 \in \states,$ a function $f\colon \states \to \R$ is said to Büchi-rank $(s_1,s_2)$ where $s_1 \tstrans s_2$, if it satisfies one of the following:
	\begin{compactitem}
		\item $s_1 \in \B \wedge \big[f(s_1) \geq 0 \Rightarrow f(s_2) \geq 0\big];$ or
		\item $s_1 \notin \B \wedge \big[f(s_1) \geq 0 \Rightarrow 0 \leq f(s_2) \leq f(s_1)-1\big].$
	\end{compactitem}
	$f$ is called a \textit{$\B$-Existential Büchi Ranking Function} ($\B$-EBRF) if it satisfies the following conditions:
	\begin{compactitem}
		\item $ \exists s_\init \in \states_\init \text{ where } f(s_{\init}) \geq 0$.
		\item For every $s_1 \in \states$, there exists $s_2 \in \states$ such that $s_1 \tstrans s_2$ and $(s_1,s_2)$ is Büchi-ranked by $f$.
	\end{compactitem}
\end{definition}


\paragraph{Example}
	The following is a $\{(l_1,q_1, \ast)\}$-EBRF for the transition system in Fig.~\ref{fig:2}: $f(l,x_0) = x_0+3$ if $l = (l_\init,q_0)$, $f(l,x_0) = x_0+2$ if $l = (l_1,q_0)$, $f(l,x_0) = x_0+1$ if $l = (l_2,q_0)$, $f(l,x_0) = 0$ if $l = (l_1,q_1)$ and $f(l,x_0) = 0$ otherwise.
    
	For example, the state $s_0 = ((l_1,q_0),1)$ has two successors in the transition system: $s_1 = ((l_2,q_0),1)$ and $s_2 = ((l_3,q_0),1)$. It is easy to see that $0 \leq f(s_1) \leq f(s_0)-1$ which shows that transition from $s_0$ to $s_1$ is Büchi-ranked by $f$.

The following theorem, proved in Appendix \ref{app:exist-proof}, establishes the soundness and completeness of EBRFs for the EB-PA problem, which is the main result of this section. Hence, since we showed in Lemma~\ref{reduction} that one can reduce the LTL-RP problem to EB-PA, as a corollary it also follows that EBRFs provide sound and complete certificates for LTL-RP.

\begin{theorem}[Soundness and Completeness of EBRFs for EB-PA]\label{thm:ebrf}
	There exists a $\B$-EBRF $f$ for $\T$ with Büchi specification $\B$ if and only if the answer to the EB-PA problem of $\T$ and $\B$ is positive.
\end{theorem}

\begin{corollary}
	\label{LTL-RP-EBRF}
	The answer to the LTL-RP problem of $\T$ and $\formula$ is positive if and only if there exists a $\B_N^\T$-EBRF for $\T \times N,$ where $N$ is the NBW accepting~$\neg \formula$.
\end{corollary}


\subsection{Sound and Complete Witnesses for Universal B-PA}
\label{Universal}

Similarly to EBRFs, we can define a witness function for the UB-PA problem. The difference compared to EBRFs is that we now impose the B\"uchi ranking condition for {\em every} successor state of a state in which the witness function is non-negative. In contrast, in EBRFs we imposed the B\"uchi ranking condition only for {\em some} successor state.

\begin{definition}[UBRF]
	\label{UBRF}
	A function $f\colon \states \to \R^n$ is called a \textit{$\B$-Universal Büchi Ranking Function} ($\B$-UBRF) if it satisfies the following conditions:
	\begin{compactitem}
		\item $f(s) \geq 0$ for \textbf{every} $s\in \states_\init$
		\item For \textbf{every} $s_1,s_2 \in \states$ such that $s_1 \tstrans s_2$,  $(s_1,s_2)$ is Büchi-ranked by $f$.
	\end{compactitem}
\end{definition}
We have the following theorem, which establishes that UBRFs provide a sound and complete certificate for the UB-PA problem. The proof is similar to the existential case and presented in Appendix~\ref{app:universal-proof}. The subsequent corollary then follows from Lemma~\ref{reduction} which shows that the LTL-VP problem can be reduced to the UB-PA problem if $\formula$ admits a deterministic B\"uchi automaton.

\begin{theorem}[Soundness and Completeness of UBRFs for UB-PA]\label{thm:ubrf}
	There exists a $\B$-UBRF $f$ for $\T$ with Büchi specification $\B$ if and only if the answer to the UB-PA problem of $\T$ and $\B$ positive.
\end{theorem}


\begin{corollary}
	If $\formula$ is an LTL formula that admits a DBW $D$, the answer to the LTL-VP problem of $\T$ and $\formula$ is positive iff there exists a $\B_D^\T$-UBRF for $\T \times D$.
\end{corollary}

\paragraph{Remark} Note that if the transition system $\T$ is deterministic, (i.e.~it contains no non-determinism in initial states, assignments or branches)
the LTL-VP of $\T$ and $\formula$ will be equivalent to the LTL-RP of $\T$ and $\neg \formula$. 
Thus, in this case, the B\"uchi automaton determinism assumption can be relaxed as follows: if $N$ is an NBW that accepts the same language as $\formula$, the answer to the LTL-VP of $\T$ and $\formula$ is positive if and only if there exists a $\B_N^T$-EBRF for $\T \times N$.  

	\section{Template-based Synthesis of Polynomial Witnesses}
\label{algo}


We now present our fully automated algorithms to synthesize polynomial EBRFs and UBRFs in polynomial transition systems. A transition system $\T$ is said to be {\em polynomial} if guards and updates of all transitions in $\T$ are polynomial expressions over program variables $\V$. Given a polynomial transition system $\T$ and a Büchi specification $\B$, which was obtained from an LTL formula as above, our approach synthesizes polynomial EBRFs and UBRFs of any desired degree, assuming that they exist.
Our algorithms follow a template-based synthesis approach, similar to the methods used for reachability and termination analysis~\cite{reach20,invariant20}. In particular, both EBRF and UBRF synthesis algorithms first fix a symbolic polynomial template function for the witness at each location in $\T$. The defining conditions of EBRFs/UBRFs are then expressed as entailment constraint of the form 
\begin{equation} \label{eq:entail} \textstyle \exists c \in \R^m ~~ \forall e \in \R^n ~~ (\phi \Rightarrow \psi),\end{equation}
where $\phi$ and $\psi$ are conjunctions of polynomial inequalities. We show that this translation is sound and complete. However, such constraints are notoriously difficult to solve due to the existence of a quantifier alternation. Thus, we use the sound and semi-complete technique of~\cite{reach20} to eliminate the quantifier alternation and translate our constraints into a system of purely existentially quantified quadratic inequalities. Finally, this quadratic programming instance is solved by an SMT solver. We note that a central technical difficulty here is to come up with sound and complete witness notions whose synthesis can be reduced to solving entailment constraints of the form~\eqref{eq:entail}. While~\cite{reach20,invariant20} achieved this for termination and reachability, our EBRF and UBRF notions significantly extend these results to arbitrary LTL formulas. 

As is common in static analysis tasks, we assume that the transition system comes with an invariant $\invariant_l$ at every location $l$ in $\T$. Invariant generation is an orthogonal and well-studied problem. In polynomial programs, invariants can be automatically generated using the tools in~\cite{aspic,invariant20,icra}. Alternatively, one can encode an inductive invariant via constraints of the form~\eqref{eq:entail}. This has the extra benefit of ensuring that we always find an invariant that leads to a witness for our LTL formula, if such a witness exists, and thus do not sacrifice completeness due to potentially loose invariants. See~\cite{invariant20} for details of the encoding. This is the route we took in our tool, i.e.~our tool automatically generates the invariants it requires using the sound and complete method of~\cite{invariant20}. For brevity, we removed the invariant generation part from the description of the algorithms below.

\paragraph{Synthesis of Polynomial EBRFs}
We now present our algorithm for synthesis of a polynomial EBRF, given a polynomial transition system $\T = (\variables, \locations, l_{\init}, \precondition, \tstrans)$ and Büchi specification $\B$ obtained from an LTL formula with polynomial inequalities in $\AP$. We present a detailed example that illustrates the steps of the algorithm in appendix \ref{app:example}. The algorithm has five steps:

\begin{compactenum}
	\item \textit{Fixing Symbolic Templates.} 
	Let $M_\V^D = \{m_1, m_2, \dots, m_k\}$ be the set of all monomials of degree at most $D$ over the set of variables $\variables$. In the first step, the algorithm generates a symbolic polynomial template for the EBRF at each location $l \in \locations$ as follows: $\textstyle f_l(x) = \Sigma_{i=1}^{k} c_{l,i} \cdot m_i.$
	Here, all the $c$-variables are fresh symbolic template variables that represent the coefficients of polynomial expressions in $f$. The goal of our synthesis procedure is to find a concrete valuation of $c$ variables for which $f$ becomes a valid $\B$-EBRF for $\T$.
	
	

	%
	\item \textit{Generating Entailment Constraints.} For every location $l \in L$ and variable valuation $x \models \invariant_l$, there must exist an outgoing transition $\tau$ such that $x \models \guard_\tau$ and $\tau$ is Büchi-ranked by $f$ in $x$. The algorithm symbolically writes down this condition as an entailment constraint: 
	$ \label{cp2} \forall x \in \R^n ~~\, x \models (\phi_l \Rightarrow \psi_l)
	$
	 with $\phi_{l}$ and $\psi_{l}$ symbolically computed as follows:
	$
	\textstyle
	\phi_{l} := \invariant_l \wedge f_l(x) \geq 0 \text{ and } 	\psi_{l} \equiv \bigvee_{\tau \in \outtrans_l } \guard_\tau \wedge \B\textrm{--}Rank(\tau),
	$
	where for each $\tau = (l,l', \guard_\tau, \update_\tau)$ the predicate $\B\textrm{--}Rank$ is defined as
	\begin{small}
	\[
	 \B\textrm{--}Rank\equiv
	\begin{small} \begin{cases}
		f_{l'}(\update_\tau(x)) \geq 0 \wedge f_{l'}(\update_\tau(x)) \leq f_l(x)-1 & l \notin \B \\
		f_{l'}(\update_\tau(x)) \geq 0	& l \in \B
	\end{cases} \end{small}
        \]
	\end{small}
    The algorithm then writes $\psi_l$ in disjunctive normal form as $\textstyle \vee_{i=1}^k \psi_{l,i}$. Next, the algorithm rewrites $\phi_l \Rightarrow \psi_l$ equivalently as:
	\begin{equation} 
	\textstyle (\phi_l \wedge \bigwedge_{i=1}^{k-1} \neg \psi_{l,i}) \Rightarrow \psi_{l,k} \label{const_pair} 
	\end{equation}
	This rewriting makes sure that we can later manipulate the constraint in~\eqref{cp2} to fit in the standard form of~\eqref{eq:entail}\footnote{We have to find values for $c$-variables that satisfy all these constraints conjunctively. This is why we have an extra existential quantifier in~\eqref{eq:entail}.}.
	Intuitively, \eqref{const_pair} ensures that whenever $l$ was reached and each of the first $k-1$ outgoing transitions were either unavailable or not Büchi-ranked by $f$, then the last transition has to be available and Büchi-ranked by $f$. Our algorithm populates a list of all constraints and adds the constraint~\eqref{const_pair} to this list before moving to the next location and repeating the same procedure.
	Note that in all of the generated constraints of the form~\eqref{const_pair}, both the LHS and the RHS of the entailment are boolean combinations of polynomial inequalities over program variables.
	
		

	\item \textit{Reduce Constraints to Quadratic Inequalities.} 
	To solve the constraints generated in the previous step, we directly integrate the technique of~\cite{reach20} into our algorithm. This is a sound and semi-complete approach based on Putinar's Positivstellensatz. We will provide an example below, but refer to~\cite{reach20} for technical details and proofs of soundness/completeness of this step. 

	In this step, for each constraint of the form $\Phi \Rightarrow \Psi$, the algorithm first rewrites $\Phi$ in disjunctive normal form as $\phi_1 \vee \dots \vee \phi_t$ and $\Psi$ in conjunctive normal form as $\Psi \equiv \psi_1 \wedge \dots \wedge \psi_r$. Then for each $1 \leq i \leq t$ and $1 \leq j \leq r$ the algorithm uses Putinar’s Positivstellensatz in the exact same way as in \cite{reach20} to generate a set of  quadratic inequalities equivalent to $\phi_i \Rightarrow \psi_j.$ The algorithm keeps track of a quadratic program $\Gamma$ and adds these new inequalities to it conjunctively.

	\item \textit{Handling Initial Conditions.} Additionally, for every variable $x \in \V$, the algorithm introduces another symbolic template variable $t_x,$ modeling the initial value of $x$ in the program, and adds the constraint $[\precondition(t) \wedge f_{\initloc}(t) \geq 0]$ to $\Gamma$ to impose that there exists an initial state in $\T$ at which the value of the EBRF $f$ is non-negative.
	%
	
	
	\item \textit{Solving the System.} Finally, the algorithm uses an external solver (usually an SMT solver) to compute values of $t$ and $c$ variables for which $\Gamma$ is satisfied. If the solver succeeds in solving the system of constraints $\Gamma$, the computed values of $c$ and $t$ variables give rise to a concrete instance of an $\B$-EBRF for $\T$. This implies that the answer to the EB-PA problem is positive, and the algorithm return ``Yes''. Otherwise, the algorithm returns ``Unknown'', as there might exist a $\B$-EBRF for $\T$ of higher maximum polynomial degree $D$ or a non-polynomial $\B$-EBRF. 
	
\end{compactenum}

\begin{theorem}[Existential Soundness and Semi-Completeness]
	\label{algorithm_results_exist}
	The algorithm above is a sound and semi-complete reduction to quadratic programming for synthesizing an EBRF in a polynomial transition system $\T$ given a  Büchi specification $\B$ obtained from an LTL formula with polynomial inequalities in $\AP.$  Moreover, for any fixed $D,$ the algorithm has sub-exponential complexity.
\end{theorem}

In the above theorem, soundness means that every solution to the QP instance is a valid EBRF and semi-completeness means that if a polynomial EBRF exists and the chosen maximum degree $D$ is large enough, then the QP instance will have a solution.
In practice, we simply pass the QP instance to an SMT solver. Since it does not include a quantifier alternation, the SMT solvers have dedicated heuristics and are quite efficient on QP instances.

\paragraph{Synthesis of Polynomial UBRFs} Our algorithm for synthesis of UBRFs is almost the same as our EBRF algorithm, except that the constraints generated in Steps 2 and 4 are slightly different.

\paragraph{Changes to Step 2} Step 2 is the main difference between the two algorithms. In this step, for each location $l \in L$ and each transition $\tau \in \outtrans_l$ the UBRF algorithm adds $(\phi_{l,\tau} \Rightarrow \psi_{l,\tau})$ to the set of constraints, where we have
$
		\phi_{l,\tau} \equiv \invariant_l \wedge \guard_\tau \wedge f_l(x) \geq 0 \text{ and }
		\psi_{l,\tau} \equiv  \B\textrm{--}Rank(\tau).
$
	The intuition behind this step is that whenever a transition is enabled, it has to be Büchi-ranked by $f$.

\paragraph{Changes to Step 4} In this step, instead of searching for a suitable initial valuation for program variables, the algorithm adds the quadratic inequalities equivalent to $(\precondition \Rightarrow f_{\initloc}(x) \geq 0)$ to $\Gamma$. The quadratic inequalities are obtained exactly as in Step 3. This is because the value of the UBRF must be non-negative on every initial state of the transition system.

In the universal case, we have a similar theorem of soundness and semi-completeness whose proof is exactly the same as Theorem~\ref{algorithm_results_exist}.

\begin{theorem}[Universal Soundness and Semi-Completeness]
\label{algorithm_results_uni}
	The algorithm above is a sound and semi-complete reduction to quadratic programming for synthesizing an UBRF in a polynomial transition system $\T$ given a Büchi specification $\B$ obtained from an LTL formula with polynomial inequalities in $\AP.$ Moreover, for any fixed maximum polynomial degree $D,$ the algorithm has sub-exponential complexity.
\end{theorem}

	\section{Experimental Results}
\label{sec:exper}

\newcommand{\RQ}[1]{\bf{RQ#1}}


\paragraph{General Setup of Experiments} We implemented a prototype \footnote{Available at \href{https:/github.com/ekgma/LTL-VerP}{github.com/ekgma/LTL-VerP}} of our UBRF and EBRF synthesis algorithms in Java and used Z3 \cite{Z3}, Barcelogic \cite{bclt} and MathSAT5 \cite{mathsat} to solve the generated systems of quadratic inequalities. More specifically, after obtaining the QP instance, our tool calls all three SMT solvers in parallel. 
We also used ASPIC~\cite{aspic} for invariant generation for benchmarks that are linear programs. Experiments were performed on a Debian 11 machine with a 2.60GHz Intel E5-2670 CPU and 6 GB of RAM with a timeout of 1800 seconds.


\paragraph{Baselines} We compare our tool with Ultimate LTLAutomizer~\cite{UltimatePaper}, nuXmv~\cite{nuxmv}, and MuVal \cite{unno2023muval} as well as with a modification of our method that instead of using Putinar's Positivstellensatz simply passes entailment constraints to the SMT-solver Z3~\cite{Z3}:
	\begin{compactitem}
		\item Ultimate LTLAutomizer makes use of ``Büchi programs'', which is a similar notion to our product of a transition system and a Büchi Automaton, to either prove that every lasso shaped path in the input program satisfies the given LTL formula, or find a path that violates it. However, in contrast to our tool, it neither supports non-linear programs nor provides completeness. 
		\item nuXmv is a symbolic model checker with support for finite and infinite transition systems. It allows both existential and universal LTL program analysis and supports non-linear programs. It does not provide any completeness guarantees.
		\item MuVal~\cite{unno2023muval} is a fixed-point logic validity checker based on pfwCSP solving \cite{unno2021csp}. It supports both linear and non-linear programs with integer variables and recursive functions. 
		\item When directly applying Z3, instead of the dedicated quantifier elimination method (Step 3 of our algorithm), we directly pass the quantified formula~\eqref{eq:entail} to the solver, which will in turn apply its own generic quantifier elimination. This is an ablation experiment to check whether Step 3 is needed in practice. 
	\end{compactitem}



\paragraph{Benchmarks} We gathered benchmarks from two sources:
\begin{compactitem}
	\item 297 benchmarks from the ``Termination of C-Integer Programs'' category of TermComp'22~\cite{termcomp}\footnote{There were originally 335 benchmarks, but we had to remove benchmarks with unbounded non-determinism and those without any variables, since they cannot be translated to transition systems and are not supported in our setting.}. Among these, 287 programs only contained linear arithmetic which is supported by all comparator tools, whereas 10 programs (Appendix ~\ref{app:non-lin}) contained polynomial expressions not supported by Ultimate.
	\item 21 non-linear benchmarks from the ``ReachSafety-Loops \texttt{nl-digbench}'' category of SV-COMP'22 \cite{svcomp}\footnote{The original benchmark set contains 28 programs, but 7 of them contain unsupported operators such as integer mod and are thus not expressible in our setting.}. As these benchmarks are all non-linear, none of them are supported by Ultimate. 
\end{compactitem} 

\paragraph{LTL specifications} We used the four LTL specifications shown in Table~\ref{formula_description}. In all four considered specifications, $x$ represents the alphabetically first variable in the input program. The motivation behind our specifications is as follows:
\begin{compactitem}
	\item {\em Reach-avoid (RA) specifications.} The first specification is an example of a reach-avoid specification, which specifies that a program run should terminate without ever making $x$ negative. Reach-avoid specifications are standard in the analysis of dynamical and hybrid systems~\cite{reach-avoid1,reach-avoid2,ZikelicLHC23}. Another example is requiring a program to termination while satisfying all program assertions. 
	\item {\em Overflow (OV) specifications.} Intuitively, we want to evaluate whether our approach is capable of detecting variable overflows. The second specification specifies that each program run either terminates or the value of the variable $x$ overflows. Specifically, suppose that an overflow is handled as a runtime error and ends the program. The negation (refutation) of this specification models the existence of a run that neither terminates nor overflows and~so~converges.
	\item {\em Recurrence (RC) specifications.} The third specification is an instance of recurrence specifications which specify that a program run visits a set of states infinitely many times~\cite{MannaP89}. Our example requires that a program run contains infinitely many visits to states in which $x$ has a non-negative value.
	\item {\em Progress (PR) specifications.} The fourth specification is an example of progress specifications. In our experimental evaluation, progress specification specifies that a program run always makes progress from states in which the value of $x$ is less than $-5$ to states in which the value of $x$ is strictly positive.
 
\end{compactitem}

\begin{table}[t]
	\centering
	\resizebox{!}{0.95cm}{
	\begin{tabular}{|c|c|c|}
		\hline
		{Name} & {Formula} & {Pre-condition $\invariant_\init$}  \\
		\hline
		$RA$ & $(F\ \at(l_\term)) \wedge (G\ x \geq 0)$ & $\forall x \in \V, 0 \leq x \leq 64$ \\
		\hline 
		$OV$ & $F\ (\at(l_\term) \vee x<-64 \vee x>63)$ & $\forall x \in \V, -64 \leq x \leq 63$ \\
		\hline
		$RC$ & $G\ F\ (x \geq 0)$ & $\forall x \in \V, -64 \leq x \leq 63$ \\
		\hline 
		$PR$ & $G\ (x<-5  \Rightarrow F\ (x>0))$ & $\forall x \in \V, -64 \leq x \leq 63$ \\
		\hline 
	\end{tabular}}
	\caption{LTL specifications used in our experiments.}
	\label{formula_description}
\end{table}

\begin{table}[t]
	\centering
	\resizebox{!}{1.95cm}
        {
	\begin{tabular}{|c|c||c|c|c|c||c|c|c|c||c|c|c|c||c|c|c|c||c|c|c|c|}
		\hline
		& \multirow{2}{*}{Formula} & \multicolumn{4}{c||}{Ours} & \multicolumn{4}{c||}{Ultimate} & \multicolumn{4}{c||}{nuXmv} & \multicolumn{4}{c||}{MuVal} & \multicolumn{4}{c|}{Z3}\\
		\cline{3-22}
		& & Yes & No & Tot. & U. & Yes & No & Tot. & U. & Yes & No & Tot. & U. & Yes & No & Tot. & U. & Yes & No & Tot. & U.\\
		\hline \hline
            \parbox[t]{3mm}{\multirow{5}{*}{\rotatebox[origin=c]{90}{Linear}}}
		& $RA$ & 141 & 114 & 255 & 5 & 142 & 121 & 263 & 7 & 76 & 91 & 137 & 0 & 118 & 76 & 194 & 0 & 56 & 36 & 92 & 0 \\
		\cline{2-22}
		& $OV$ & 199 & 47 & 246 & 4 & 212 & 55 & 267 & 5 & 110 & 50 & 160 & 0 & 205 & 47 & 252 & 3 & 48 & 27 & 75 & 0\\
		\cline{2-22}
		& $RC$ & 87 & 187 & 274 & 0 & 86 & 194 & 280 & 0 & 83 & 183 & 266 & 0 & 86 & 191 & 277 & 0 & 44 & 71 & 115 & 0\\
		\cline{2-22}
		& $PR$ & 43 & 222 & 265 & 1 & 45 & 237 & 282 & 0 & 44 & 227 & 271 & 0 & 42 & 235 & 277 & 0 & 29 & 77 & 106 & 0 \\
            \hhline{~=====================|}
		& Avg. T & 5.4 & 81.5 & 47.2 & - & 5.4 & 4.1 & 4.7 & - & 248.9 & 13.5 & 98.7 & - & 48.8 & 8.43 & 26.4 & - & 18.5 & 160.6 & 95.7 & - \\
		\hline 
            \hline 
             \parbox[t]{3mm}{\multirow{5}{*}{\rotatebox[origin=c]{90}{Non-linear}}}
		& $RA$ & 24 & 3 & 27 & 8 & - & - & - & - & 1 & 0 & 1 & 0 & 18 & 1 & 19 & 2 & 0 & 0 & 0 & 0 \\
		\cline{2-22}
		& $OV$ & 26 & 0 & 26 & 2 & - & - & - & - & 7 & 0 & 7 & 0 & 25 & 0 & 25 & 1 & 0 & 0 & 0 & 0\\
		\cline{2-22}
		& $RC$ & 20 & 6 & 26 & 0 & - & - & - & - & 17 & 9 & 26 & 2 & 17 & 7 & 24 & 2 & 0 & 0 & 0 & 0\\
		\cline{2-22}
		& $PR$ & 11 & 16 & 27 & 1 & - & - & - & - & 9 & 16 & 25 & 0 & 5 & 16 & 21 & 1 & 0 & 0 & 0 & 0\\
		\hhline{~=====================|}
		& Avg. T & 10.7 & 99.1 & 32.3 & - & - & - & - & - & 34.6 & 0.3 & 20.0 & - & 109.6 & 14.7 & 84.7 & - & - & - & - & - \\
            \hline
	\end{tabular}}
	\caption{Summary of our experimental results. For each class of benchmarks (linear/non-linear) and each formula, We report in how many cases the tool could successfully prove the formula (Yes) or refute it (No), total number of cases proved by the tool (Tot.), number of instances uniquely solved by each tool and no other tools (U.), and average runtime of each tool on programs that were successfully proved as correct with respect to each specification (Avg. T).}
	\label{experiments_results_linear}
\end{table}

\paragraph{Results on Linear Programs}
The top rows of Table \ref{experiments_results_linear} summarize our results over linear benchmarks to which all tools are applicable. First, we observe that in all cases our tool outperforms the method that uses Z3 for quantifier elimination, showing that our Step 3 is a crucial and helpful part of the algorithm. Compared to nuXmv, our tool proves more instances in all but two LTL refutation and one LTL verification cases, i.e.~the ``No'' column for the OV and PR specifications and the ``Yes'' column for the PR specification. On the other hand, our prototype tool is on par with Ultimate and MuVal, while proving 10 unique instances. Note that Ultimate is a state of the art and well-maintained competition tool that is highly optimized with heuristics that aim at the linear case. In contrast, it cannot handle polynomial instances. Our results shown in Table~\ref{experiments_results_linear} demonstrate that our prototype tool is very competitive already on linear benchmarks, even though our main contribution is to provide practically-efficient semi-complete algorithms for the polynomial case.

\paragraph{Unique Instances} An important observation is that our tool successfully handles 10 unique \emph{linear} instances that no other tool manages to prove or refute. Thus, our evaluation shows that our method handles not only polynomial, but even linear benchmarks that were beyond the reach of the existing methods. This shows that our algorithm, besides the desired theoretical guarantee of semi-completeness, provides an effective automated method. Future advances in invariant generation and SMT solving will likely further improve the performance.

\paragraph{Runtimes} Our tool and Ultimate are the fastest tools for proving LTL verification instances with an equal average runtime of 5.4 seconds. For LTL refutation, our tool is slower than other tools. 

\paragraph{Results on Non-Linear Programs}	
The bottom rows of Table~\ref{experiments_results_linear} show the performance of our tool and the baselines on the non-linear benchmarks. Ultimate does not support non-linear arithmetic and Z3 timed out on every benchmark in this category. Here, compared to nuXmv, our tool succeeded in solving strictly more instances in all but one formula, i.e. $RC,$ where both tools solve the same number of instances. In comparison with MuVal, our tool proves more instances for all four formulas. Moreover, the fact that Z3 timed out for every program in this table is further confirmation of the practical necessity of Step~3 (Quantifier Elimination Procedure of~\cite{reach20}) in our algorithm. Note that our prototype could prove 11 instances that none of the other tools could handle. 


\paragraph{Summary} Our experiments demonstrate that our automated algorithms are able to synthesize both LTL verification and refutation witnesses for a wide variety of programs. Our technique outperforms the previous methods when given non-linear polynomial programs (Bottom rows of Table~\ref{experiments_results_linear}). Moreover, even in the much more widely-studied case of linear programs, we are able to handle instances that were beyond the reach of previous methods and to solve the number of instances that is close to the state-of-the-art tools (Top Rows of Table~\ref{experiments_results_linear}).

	\section{Conclusion} \label{sec:conclu}




We presented a novel family of sound and complete witnesses for template-based LTL verification.
Our approach is applicable to both verification and refutation of LTL properties in programs. It unifies and significantly generalizes previous works targeting special cases of LTL, e.g.~termination, safety and reachability. We also showed that our LTL witnesses can be synthesized in a sound and semi-complete manner by a reduction to quadratic programming. Our reduction works when the program and the witness are both polynomial.
An interesting direction of future work would be to consider non-numerical programs that allow heap-manipulating operations. A common approach to handling heap-manipulating operations is to construct numerical abstractions of programs~\cite{DBLP:journals/fmsd/BouajjaniBHIMV11,DBLP:conf/popl/MagillTLT10} and perform the analysis on numerical abstractions. Thus, coupling such approaches, e.g.~\cite{CookK13}, with our method is a compelling future direction.

     \paragraph{Acknowledgements} This work was supported in part by the ERC-2020-CoG 863818 (FoRM-SMArt) and the Hong Kong Research Grants Council ECS Project Number 26208122.

	\bibliographystyle{splncs04}
	\bibliography{sample-base}
	\appendix
	\newpage

\section{LTL Semantics}
\label{app:sem}
\paragraph{Semantics} Each atomic proposition $p$ in $\AP$ determines a set of states in $\T$ at which $p$ is satisfied. We write $(l,e)\models p$ to denote such satisfaction. Specifically, if $p$ is of the form $\textrm{exp}(\mathbf{x})\geq 0$ then $(l,e)\models p$ if $\textrm{exp}(e)\geq 0$, and if $p$ is of the form $\at(l')$ then $(l,e)\models p$ if $l=l'$. Our LTL semantics are then defined in the usual way. Given an infinite sequence $\run = (l_0,e_0), (l_1,e_1), \dots $ of states, we write $ \run \models \formula$ to denote that it satisfies $\formula$ and have:
\begin{compactitem}
	\item For every $p \in \AP,$ $\run \models p$ iff $(l_0,e_0)\models p$;
	\item $\run \models \X\ \formula$ iff $\run^{1+} \models \formula$;
	\item $\run \models \F\ \formula$ iff $\run^{i+} \models \formula$ for some $i \in \N_0$;
	\item $\run \models \G\ \formula$ iff $\run^{i+} \models \formula$ for all $i \in \N_0$;
	\item $\run \models \formula\ \U\ \psi$ iff there exists $k \in \N_0$ s.t.~$\run^{i+} \models \formula$ for all $0\leq i < k$ and, additionally, $\run^{k+} \models \psi.$
\end{compactitem}

\section{From LTL Formulas to Büchi Automata}
\label{app:LTL2BA}
It is a classical result in model checking that for every LTL formula $\formula$ defined over atomic predicates $\AP$ there exists a non-deterministic Büchi automaton $N$ with alphabet $\powerset{AP}$ which accepts exactly those traces that satisfy $\formula$~\cite{henzinger2018handbook}. Here, $\powerset{\AP}$ is the set of subsets of $\AP.$ 

We now define the product of a transition system $\T$ and an NBW $N$ over the alphabet $\powerset{\AP}.$ Note that at every step of a run of a transition system, we are always at a unique location. Hence, if a letter $\alpha \in \powerset{\AP}$ contains both $\at(l)$ and $\at(l')$ for two different locations $l \neq l',$ then it will never be realized by any run. Thus, without loss of generality, we can remove any transitions with such an $\alpha$. A similar argument applies if $\alpha$ has no atomic proposition of the form $\at(\cdot)$\footnote{If $\delta(q, \alpha) = \emptyset,$ i.e.~if the transition is removed, then any run that is at $q$ and reads the letter $\alpha$ will be a rejecting run. Equivalently, we can create a non-accepting sink state $\perp$ and let $\delta(q, \alpha) = \perp.$ The sink state will only have self-loops, i.e.~$\delta(\perp, \beta) = \perp$ for every $\beta \in \powerset{\AP}.$}.

\paragraph{Product of a Transition System and a Büchi Automaton}
	Let $\T = (\variables, \locations,  \initloc, \precondition, \tstrans)$ be a transition system and $N = (\nbastates,\powerset{\AP}, \nbatrans, \nbainit, \nbaaccept)$ an NBW. The {\em product} of $\T$ and $N$ results in a new transition system $\T \times N := (\variables, \locations', \initloc', \precondition, \leadsto)$ and a set of states $\B^\T_N$ in $\T \times N$ where:
    \begin{compactitem}
        \item $\locations'$ is the Cartesian product $\locations \times \nbastates$.
        \item $\initloc' = (\initloc, q_0)$.
        \item For each transition $\tau=(l,l',\guard_\tau,\update_\tau) \in \,\tstrans$, states $q,q' \in \nbastates$ and $\alpha\in \powerset{\AP}$ such that $q' \in \nbatrans(q,\alpha)$ and $\at(l)$ is the unique proposition of the form $\at(\cdot)$ in $\alpha$, we include a transition $\big((l,q), (l',q'), \guard_\tau \wedge \alpha^{-\at}, \update_\tau \big)$ in $ \leadsto$. Here, $\alpha^{-\at}$ is the conjunction of all atomic propositions in $\alpha$ except $\at(l).$  
        \item $\B^\T_N = \locations \times \nbaaccept \times \R^n$, i.e.~$\B^\T_N$  is the set of states in $\T \times N$ whose NBW state component is contained in $\nbaaccept$.
    \end{compactitem}

\section{Proof of Lemma \ref{reduction}}
\label{proof:lm1}


\subsection{From LTL-RP to EB-PA}
\begin{proof}
	First, suppose that there exists a 
 run $\run = (l_0, e_0), (l_1, e_1), \dots$ in $\T$ that satisfies $\formula$. For each $i$ let $\alpha_i\in \powerset{\AP}$ be the set of all atomic predicates that are true in $(l_i,e_i)$. Then, since $N$ accepts the same language as $\formula$, we know that $\alpha_0,\alpha_1,\dots$ is an infinite word in the alphabet $\powerset{\AP}$ which is accepting in $N$. Let $q_0,q_1,\dots$ be an accepting run in $N$ induced by this infinite word, and consider the following sequence of states in $\T \times N$:
	\[
	\run' = (l_0, q_0, e_0), (l_1, q_1, e_1), \dots
	\]
	We claim that this sequence is a run, i.e.~that $(l_i, q_i, e_i) \leadsto (l_{i+1}, q_{i+1}, e_{i+1})$ for each $i$ where $\leadsto$ is the set of transitions in $\T \times N$. To see this, note that since $\run = (l_0, e_0), (l_1, e_1), \dots$ is a run in $\T$, for each $i$ there exists a transition $\tau=(l_i,l_{i+1},\guard_\tau,\update_\tau)$ under which $(l_{i+1},e_{i+1})$ is a successor of $(l_i,e_i)$. On the other hand, we also know that $q_{i+1}\in \nbatrans(q_i,\alpha_i)$ where $\nbatrans$ is the transition function of $N$. Finally, since $\alpha_i$ is the set of all atomic predicates in $\AP$ that are true at $(l_i,e_i)$, we must have $\at(l_i)\in\alpha_i$. Hence, we have that $(l_i, q_i, e_i) \leadsto (l_{i+1}, q_{i+1}, e_{i+1})$ under transition $((l_i,q_i),(l_{i+1},q_{i+1}),\guard_\tau\land \alpha^{-\at},\update_\tau)$ and the claim follows. Finally, since $q_0,q_1,\dots$ is accepting in $N$ we have that $\run' = (l_0, q_0, e_0), (l_1, q_1, e_1), \dots$ contains infinitely many states in $\B^\T_N$.
	
	For the opposite direction, suppose that there exists a run $\run' = (l_0, q_0, e_0), (l_1, q_1, e_1), \dots$  in $\T \times N$ that visits $\B^\T_N$ infinitely many times. The definition of transitions in the product transition system $\T \times N$ implies that $(l_0, e_0), (l_1, e_1), ...$ is a run in $\T$, that $q_0,q_1,\dots$ is an accepting run in $N$ and that there exists an infinite accepting word $\alpha_0,\alpha_1,\dots$ in $\powerset{AP}$ such that $\at(l_i) \in \alpha_i$ for each $i$. Since $\alpha_0,\alpha_1,\dots$ is accepting in $N$ and $\at(l_i) \in \alpha_i$ for each $i$ and since $N$ and $\formula$ accept the same language, it follows that $(l_0, e_0), (l_1, e_1), ...$ satisfies $\formula$.  \hfill\qed
\end{proof}
\subsection{From LTL-VP to UB-PA} 
\begin{proof}
	First suppose that every run $\run$ of $\T$ satisfies $\formula$. Let $\run' = (l_0,e_0,q_0), (l_1,e_1,q_1), \dots$ be a run in $\T \times D$, then $\eta = (l_0, e_0), (l_1,e_1), \dots$ is a run in $\T$, hence satisfying $\formula$. As $D$ is deterministic, $q_0, q_1, \dots$ is the unique run of $D$ corresponding to $\eta$, so it has to be accepting. This shows that $\run'$ is $\B_D^T$-Büchi. Therefore, every run of $\T \times D$ is $\B_D^T$-Büchi and the answer to the UB-PA problem is positive.
	
	Now, suppose that every run of $\T \times D$ is $\B_D^T$-Büchi. Let $ \pi = (l_0, e_0), (l_1,e_1), \dots$ be a run in $\T$ and let $q_0, q_1, \dots$ be the corresponding run in $D$. From the definition of product transition system, it is followed that $(l_0,e_0, q_0), (l_1,e_1,q_1), \dots$ is a run of $\T \times D$ and therefore $\B_D^\T$-Büchi. This means that $q_0,q_1, \dots$ is an accepting run in $D$, showing that $\pi$ satisfies $\formula$. Therefore, every run of $\T$ satisfies $\formula$ and the answer to the LTL-VP problem is positive.\hfill\qed
\end{proof}

\section{Scheduler Definition}
\label{app:scheduler}
We define scheduler in this section and will use it in the proofs of Appendices~\ref{app:exist-proof}, \ref{appendix-mem-less-angelic}, and \ref{app:universal-proof}.

\paragraph{Schedulers} We resolve non-determinism in transition systems using the standard notion of schedulers. A {\em scheduler} is a function $\scheduler$ that chooses the initial state in $\initstates$ in which the program execution should start, and then for each state specifies the successor state to which the program execution should proceed. Formally, a scheduler is a function  $\scheduler \colon \states^* \to \states$ such that $\scheduler(\Lambda) \in \initstates$ and $s \tstrans \scheduler(\run \cdot s)$ for every sequence of states $\run$ and state $s$. We use $\Sigma$ to denote the set of all schedulers in $\T$.
A scheduler is {\em memory-less}  if its choice of a successor state only depends on the final state in the sequence and does not depend on prior history, i.e.~if $\scheduler(\run \cdot s) = \scheduler(\run' \cdot s)$ for every state $s$ and pair of sequences $\run$ and $\run'$.
A scheduler $\scheduler$ naturally induces a run $\run_\scheduler$ in the transition system, defined as $\run_\scheduler(0) = \scheduler(\Lambda)$ and $\run_\scheduler(i) = \scheduler(\run_\scheduler^i)$ for every $i \in \N$. For a subset $\B \subseteq \states$ of state, a scheduler is $\B-$Büchi if its run is.

\section{Proof of Theorem \ref{thm:ebrf}} \label{app:exist-proof}
Please see Appendix~\ref{app:scheduler} for the definition of schedulers.
\begin{proof}[Soundness]
	Let $e_\init \in \theta_\init$ be such that $f(\initloc,e_\init) \geq 0$ and $(l,e) \in \states$ be any state of $\T$. By definition of EBRFs, there exists $(l',e') \in \states$ such that $(l,e) \tstrans (l',e')$ and $f$ Büchi-ranks $\big((l,e),(l',e')\big)$. We define the memory-less scheduler $\scheduler$ as follows: $\scheduler_f(\Lambda) = (\initloc, e_\init) $ and $\scheduler_f(l,e) = (l',e')$. Next we prove that $\scheduler_f$ is $\B$-Büchi by showing that $\run = \run_{\scheduler_f}$ is $\B$-Büchi. 

	As $f$ is an EBRF, $f(\run(0)) \geq 0$ and it can be inductively proved that $f(\run(i)) \geq 0$ for each $i \in \N$. Also, if $\run(i) \notin \B$, then $f(\run(i+1)) \leq f(\run(i))-1$.
	Suppose, for the sake of contradiction, that $\run$ reaches $\B$ only finite number of times. Then there exists $k \in \N$ such that $\run^{k+}$ does not reach $\B$ at all, meaning that the value of $f$ is non-negative at $\run(k)$ and decreases by 1 infinitely while remaining non-negative. The contradiction shows that $\run$ is $\B$-Büchi. \hfill\qed
\end{proof}

\begin{proof}[Completeness]
	Suppose that there exists a $\B$-Büchi scheduler $\scheduler$ for $\T$. To prove completeness, we first observe that if there exists a $\B$-Büchi scheduler $\scheduler$ for $\T$, then there also exists a memory-less $\B$-Büchi scheduler for $\T$. We defer the proof of this claim to  Appendix~\ref{appendix-mem-less-angelic}. 
	
	Let $\scheduler$ be a memory-less $\B$-Büchi scheduler for $\T$. It gives rise to a $\B$-EBRF $f$ for $\T$ as follows: Let $\run_\scheduler$ be the run in $\T$ induced by $\scheduler$. Since $\scheduler$ is $\B$-Büchi, the run $\run_\scheduler$ must visit states in $\B$ infinitely many times. Thus, for every $i \in \N_0$, let $\dist_{\run_\scheduler}(i)$ be the smallest number $k\in\mathbb{N}_0$ such that $\run_\scheduler(i+k) \in \B$. We define $f$ as follows:
    \[ 
	f(l,e) = \begin{small}\begin{cases}
		-1, &\text{if } (l,e) \notin \run_\scheduler\\
		\dist_{\run_\scheduler}(i), &\text{if } (l,e) = \run_\scheduler(i) \text{ for some }i \in \N_0.
	\end{cases}\end{small}
    \]
	In other words, $f$ is equal to $-1$ in every state that is not contained in the run $\run_\scheduler$, and for every state contained in $\run_\scheduler$ we define $f$ to be equal to the number of subsequent states in $\run_\scheduler$ before the next visit to $\B$. Note that, if for some $i, j \in \N_0$ we have $\run_\scheduler(i) = \run_\scheduler(j)$, given that the scheduler is memory-less, $\dist_{\run_\scheduler}(i) = \dist_{\run_\scheduler}(j)$, so function $f$ is well-defined. One can easily check by inspection that this function satisfies the defining conditions of EBRFs. \hfill \qed
\end{proof}

\section{Sufficiency of Memory-less Schedulers}
Please see Appendix~\ref{app:scheduler} for the definition of schedulers.
\label{appendix-mem-less-angelic}
The scheduler generated in the proof of Theorem \ref{thm:ebrf} is memory-less. We now show that the existence of memory-less schedulers that generate $\B$-Büchi runs is equivalent to the existence of generic (history dependent) schedulers that generate $\B$-Büchi runs. We start by defining the cycle-decomposition of a trajectory. Trajectories are finite sub-sequences of runs and cycles are trajectories starting and ending in the same state. We denote the trajectory containing the first $k$ states in $\run$ by $\run^k$ and write $\run^{k+}$ for the remaining suffix of the run.

\begin{definition}[Cycle-Decomposition of a Trajectory]
	For a trajectory $\tau$ starting at $s_1$ and ending at $s_2$, we define its Cycle-Decomposition as a tuple $(C, \beta)$ where $C$ is the set of all cycles appearing as a sub-trajectory of $\tau$ and $\beta$ is the underlying non-cyclic trajectory in $\tau$ going from $s_1$ to $s_2$.  See Figure \ref{cycle-dec-example} for an example.
\end{definition}

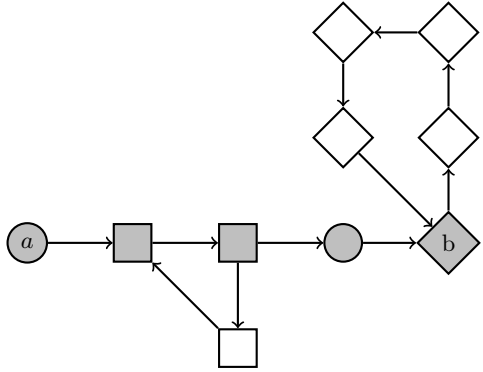
\begin{figure}
	\centering
	\begin{tikzpicture}[ scale=.7, ->]
		\node[draw,fill=lightgray, circle, minimum size=5mm, thick] (a) at (0,2) {$a$};
		\node[draw, fill=lightgray, rectangle, minimum size=5mm, thick] (1) at (2,2) {};
		\node[draw, rectangle, minimum size=5mm, thick] (2) at (4,0) {};
		\node[draw, fill=lightgray, rectangle, minimum size=5mm, thick] (3) at (4,2) {};
		\node[draw, fill=lightgray, circle, minimum size=5mm, thick] (4) at (6,2) {};
		\node[draw, fill=lightgray, diamond, minimum size=8mm, thick] (b) at (8,2) {b};
		\node[draw,diamond, minimum size=8mm, thick] (6) at (8,4) {};
		\node[draw,diamond, minimum size=8mm, thick] (7) at (8,6) {};
		\node[draw,diamond, minimum size=8mm, thick] (8) at (6,6) {};
		\node[draw,diamond, minimum size=8mm, thick] (9) at (6,4) {};
		\draw[thick] (a) to (1);
		\draw[thick] (2) to (1);
		\draw[thick] (3) to (2);
		\draw[thick] (1) to (3);
		\draw[thick] (3) to (4);
		\draw[thick] (4) to (b);
		\draw[thick] (b) to (6);
		\draw[thick] (6) to (7);
		\draw[thick] (7) to (8);
		\draw[thick] (8) to (9);
		\draw[thick] (9) to (b);
		
		
	\end{tikzpicture}
	\caption{An example finite path between two states $a$ to $b$ with its cycle-decomposition $(C, \beta)$, where square and diamond states specify the cycles in $C$ and the gray-scaled states specify $\beta$.}
	\label{cycle-dec-example}
\end{figure}

In order to construct the cycle-decomposition $(C,\beta)$ of $\tau$, we go through $\tau$ and push the states one by one into a stack. Whenever we visit a state $s$ already in the stack, we pop all the elements between the two occurrences of $s$ (including only one of the occurrences) and add them as a cycle to $C$. Finally, $\beta$ will be the trajectory constructed by the states remaining in the stack at the end.

\begin{lemma}
	There exists a \emph{memory-less} $\B$-Büchi scheduler for $\T$ \emph{if and only if} there exists a $\B$-Büchi scheduler for $\T$. 
\end{lemma}
\begin{proof}
	Necessity is trivial. For sufficiency, suppose $\scheduler$ is a $\B$-Büchi scheduler. We construct a memory-less scheduler $\scheduler'$ based on $\scheduler$. 
	Let $\run = \run_\scheduler$ and let $(C_i, \beta_i)$ be the cycle-decomposition of $\run^i$ for all $i \geq 0$. One of the following cases can happen:
	\begin{enumerate}
		\item $\run$ has no cycles. In this case, $\scheduler$ was not using the history of the run to generate $\run$. Therefore, we can define $\scheduler'$ to be the same as $\scheduler$ on $\run$ and arbitrary on other states. $\run_{\scheduler'}$ will then be the same as $\run_\scheduler$ proving that $\scheduler$ is $\B$-Büchi. 
		\item $\run$ has a finite number of cycles. So, there exists $k$ such that $\run^{k+}$ is non-cyclic which means $\beta_k \cdot \run^{k+}$ is a $\B$-Büchi run with no cycles where we can apply case (i).
		\item $\run$ has an infinite number of cycles. Here we would have two sub-cases:
		\begin{itemize}
			\item $\run$ contains a cycle $c$ starting at $\run(b)$ which intersects with $\B$. Let $k$ be the smallest index where $\run(k) = \run(b)$, then the run $\run' = \beta_k \cdot c^\omega$ is easily seen to be $\B$-Büchi. As each state in $\run'$ has a unique successor in the run, there exists a memory-less scheduler $\scheduler'$ where $\run_{\scheduler'} = \run'$.
			\item None of the cycles of $\run$ intersect with $\B$. Removing all cycles from $\run$ produces a non-cyclic $\B$-Büchi run $\run'$. The rest was handled in case (i).
		\end{itemize}
	\end{enumerate}\hfill
\end{proof}

\section{Soundness and Completeness of Universal Witnesses}
Please see Appendix~\ref{app:scheduler} for the definition of schedulers.
\label{app:universal-proof}
\begin{lemma}
	Every memory-less scheduler of $\T$ is $\B$-Büchi, if and only if every scheduler of $\T$ is $\B$-Büchi
\end{lemma}
\begin{proof}
	Necessity is trivial. For sufficiency, suppose every memory-less scheduler of $\T$ is $\B$-Büchi and let $\scheduler$ be a non-$\B$-Büchi scheduler for $\T$. Let $\run = \run_\scheduler$ and let $(C_i, \beta_i)$ be the cycle-decomposition of $\run^i$ for all $i \geq 0$. One of the following cases happens:
	\begin{enumerate}
		\item $\run$ has no cycles. In this case, $\scheduler$ does not use the history of the run to generate $\run$. Therefore, we can define $\scheduler'$ to be the same as $\scheduler$ on $\run$ and arbitrary on other states. It is then easy to see that $\run_\scheduler = \run_{\scheduler'}$ which means $\scheduler'$ is a non-$\B$-Büchi memory-less scheduler which is a contradiction. 
		\item $\run$ has a finite number of cycles. In this case, there exists $k$ such that $\run^{k+}$ is non-cyclic and does not visit $\B$. Therefore, $\beta_k \cdot \run^{k+}$ is a non-cyclic non-$\B$-Büchi run of $\T$. Similar to the previous case, this yields a contradiction.
		\item $\run$ has an infinite number of cycles. As $\run$ reaches $\B$ only a finite number of times, there must exists a cycle in $\run$ that does not reach $\run$. Suppose $c = [\run(i), \run(i+1), \dots, \run(j)]$ is one such cycle. Then $\beta_{i-1} \cdot c^{\omega}$ is a non-cyclic non-$\B$-Büchi run of $\T$. Similar to the previous cases, this yields a contradiction.
	\end{enumerate}
\end{proof}

\paragraph{Proof of Theorem~\ref{thm:ubrf}} We are now ready to prove the soundness and completeness theorem for UBRFs.

\begin{proof}[Soundness] Let $\T$ be a transition system, $f$ a $\B$-UBRF for $\T$ and $\run$ a run in $\T$. By the definition of UBRFs, $f(\run(i)) \geq 0$ for all $i \geq 0$ and $f(\run(i+1)) \leq f(\run(i))-1$ if $\run(i) \notin \B$. Now suppose for the sake of contradiction that $\run$ is not $\B$-Büchi. Then there exists $k$ such that $\run^{k+}$ does not reach $\B$ at all. So, by taking a transition from $\run(k+i)$ to $\run(k+i+1)$ the value of $f$ decreases by at least 1 while staying non-negative which yields a contradiction, showing that $\run$ is $\B$-Büchi. This proves the soundness of the witness. 
\end{proof}

\begin{proof}[Completeness] For completeness, we must show that whenever every run of $\T$ is $\B$-Büchi, there exists a $\B$-UBRF $f$ for $\T$. For a states $s \in \states$ let $d(s): \states \to \N_0$ be the defined as follows:
	\[
	d(s) = \begin{cases}
		0 & s \in \B \wedge s \text{ is reachable}\\
		\sup\limits_{s \tstrans s'} d(s')+1 & s \notin \B \wedge s \text{ is reachable} \\
		-1 & \text{otherwise}
	\end{cases}
	\label{d_definition} \tag{2}
	\]
	It is trivial that if $d$ is definable, then it is a $\B$-UBRF for $\T$. So, we show that $d$ is a well-defined. $d(s)$ is trivially well-defined in the first and last case of (\ref{d_definition}). Suppose for a moment that the co-domain of $d$ is $\N_0 \cup \{\infty\}$. We show that the $\infty$ case will never happen. For the sake of contradiction assume that there exists a reachable state $s_0 \notin B$ such that $\sup\limits_{s_0 \tstrans s'} d(s')+1 = \infty$. So, there must exist $s_1$ such that $s_0 \tstrans s_1$ and $d(s_1) = \infty$. Continuing the same procedure inductively, there must exist $s_i$ such that $s_{i-1} \tstrans s_i$ and $d(s_i) = \infty$. Hence, if $\alpha$ would be the trajectory ending at $s_0$, the run $\alpha, s_1, s_2, \dots$ is not $\B$-Büchi, which is a contradiction. So, $d$ is well-defined.
\end{proof}

\section{Proof of Theorem \ref{algorithm_results_exist}}
\begin{proof}
	Steps 1 and 2 of the algorithm are clearly sound and complete since they simply encode the EBRF conditions in a specific equivalent format. The same applies to the initial conditions in Step 4. See~\cite{reach20} for details of Step 3 and its soundness and semi-completeness. We are using the technique of~\cite{reach20} as a black box in our Step 3 and hence the same arguments apply in our case. Our algorithm inherits semi-completeness and its dependence on polynomial degrees from~\cite{reach20}.
	
	For any fixed degree of the template polynomials, our algorithm above provides a PTIME reduction from the problem of synthesizing EBRFs/URBFs to Quadratically-constrained Quadratic Programming (QP). It is well-known that QP is solvable in sub-exponential time~\cite{grigor1988solving}. Thus, the complexity of our approach is sub-exponential, too.\hfill\qed
\end{proof}

\section{Running Example of the Algorithm}
\label{app:example}

\paragraph{Fixing Symbolic Templates} Consider the transition system in Figure \ref{fig:2} and suppose $D=2,$ i.e.~the goal of the algorithm is to generate a polynomial EBRF of degree at most 2 for the given transition system. Table~\ref{table:template} shows the template generated for each location. Note that the $c_{l,i}$ variables are treated as unknowns and the goal of the algorithm is to find suitable valuations for them so that they create a $\B$-EBRF.

	\begin{table}[t]
	\begin{center}
	\begin{tabular}{|c|c|}
		\hline
		Location & Template \\
		\hline 
		$(l_\init,q_0)$ & $c_{\init,0,0} + c_{\init,0,1} \cdot x_0 + c_{\init,0,2} \cdot x_0^2$ \\
		$(l_1,q_0)$ & $c_{1,0,0} + c_{1,0,1} \cdot x_0 + c_{1,0,2} \cdot x_0^2$ \\
		$(l_1, q_1)$ & $c_{1,1,0} + c_{1,1,1} \cdot x_0 + c_{1,1,2} \cdot x_0^2$ \\
		\hline
	\end{tabular}
	\begin{tabular}{|c|c|}
		\hline
		Location & Template \\
		\hline 
				$(l_2,q_0)$ & $c_{2,0,0} + c_{2,0,1} \cdot x_0 + c_{2,0,2} \cdot x_0^2$ \\
		$(l_3,q_0)$ & $c_{3,0,0} + c_{3,0,1} \cdot x_0 + c_{3,0,2} \cdot x_0^2$ \\
		$(l_t, q_0)$ & $c_{t,0,0} + c_{t,0,1} \cdot x_0 + c_{t,0,2} \cdot x_0^2$ \\ \hline
	\end{tabular}
		\caption{The template EBRF generated for the transition system in Figure \ref{fig:2}}
				\label{table:template}
		\end{center}
	\end{table}

\paragraph{Generating Entailment Constraints} We give the entailment constraints for location $(l_2,q_0)$. Entailment constraints from other locations can be derived similarly. \\
	Let $\invariant_{(2,0)} \equiv x_0 \geq 0$ be the invariant for $(l_2,q_0)$. The algorithm symbolically computes the following:
	\begin{equation}
	\label{example2}
	\hspace{-2em}\begin{matrix}
		\big[ x_0 \geq 0 \wedge c_{2,0,0} + c_{2,0,1} \cdot x_0 + c_{2,0,2}  \cdot  x_0^2 \geq 0 \big] \Rightarrow \\
		
		 \big[ c_{1,1,0} + c_{1,1,1}  \cdot (x_0^2+1) + c_{1,1,2} \cdot (x_0^2+1)^2 \geq 0 ~~\wedge  \\
		 c_{1,1,0} + c_{1,1,1}  \cdot  (x_0^2+1) + c_{1,1,2}  \cdot  (x_0^2+1)^2 \leq c_{2,0,0} + c_{2,0,1}  \cdot  x_0 +  \cdot c_{2,0,2}  \cdot  x_0^2-1 \big] 
	\end{matrix}
	\end{equation}
	This is to formulate that whenever a state $s=(l_2,q_0,x_0)$ is reached and $f(s)$ is non-negative, there exists a successor state $s'$ such that $0 \leq f(s') \leq f(s)-1$. As $(l_2,q_0)$ has only one successor in $\T$, this has the same format as $(\ref{const_pair})$.

 \paragraph{Reduce Constraints to Quadratic Inequalities}  The algorithm symbolically computes the following entailments as $\phi_i \Rightarrow \psi_j$ expressions from (\ref{example2}):
	\begin{equation*}
		\begin{split}
			(i)~~x_0 \geq 0 \wedge c_{2,0,0} &+ c_{2,0,1} \cdot x_0 + c_{2,0,2} \cdot x_0^2 \geq 0 \Rightarrow \\ c_{1,1,0} &+ c_{1,1,1} \cdot (x_0^2+1) + c_{1,1,2} \cdot (x_0^2+1)^2 \geq 0\\
			(ii)~x_0 \geq 0 \wedge c_{2,0,0} &+ c_{2,0,1} \cdot x_0 + c_{2,0,2} \cdot x_0^2 \geq 0 \Rightarrow \\
			 c_{1,1,0} &+ c_{1,1,1} \cdot (x_0^2+1) + c_{1,1,2} \cdot (x_0^2+1)^2 \leq \\& c_{2,0,0} + c_{2,0,1} \cdot x_0 + c_{2,0,2} \cdot x_0^2-1 \\
		\end{split}
	\end{equation*}

	We show how the first constraint above is handled. The goal is to write the RHS of (i) as the sum of several non-negative expressions derived from multiplying the LHS expressions by sum-of-squares (SOS) polynomials. The idea is that SOS polynomials are always non-negative and the polynomials on the LHS are also assumed to be non-negative, so if we can combine them to achieve the RHS, then the constraint is proven\footnote{This method is not only sound, but also complete. This is due to Putinar's Positivstellensatz~\cite{putinar1993positive}.  See~\cite{reach20} for details.}. To do this, the algorithm generates three template SOS polynomials. For simplicity, suppose they are of the form $\lambda_{i,0} + \lambda_{i,1} \cdot x_0^2$ for $i=0,1,2$ where all $\lambda_{i,j}$'s are non-negative. Next, it symbolically computes the following equality:
	\begin{small}
	\begin{equation}
		\label{putinared}
		\begin{split}
		& (\lambda_{0, 0} + \lambda_{0,1} \cdot x_0^2)  \\
		+~ & (\lambda_{1, 0} + \lambda_{1,1} \cdot x_0^2) \cdot (x_0)\\
		+~ & (\lambda_{2, 0} + \lambda_{2,1} \cdot x_0^2) \cdot (c_{2,0,0} + c_{2,0,1} \cdot x_0 + c_{2,0,2} \cdot x_0^2)\\
		=~ & c_{1,1,0} + c_{1,1,1} \cdot (x_0^2+1) + c_{1,1,2} \cdot (x_0^2+1)^2
		\end{split}
	\end{equation}
	\end{small}
	The polynomial equality in (\ref{putinared}) has to be satisfied for all values of the program variable $x_0$. This means the monomials $x_0^4, x_0^3, x_0^2, x_0, $ and $1$ should have the same coefficients on both sides. Equating the coefficients, the algorithm obtains the following quadratic constraints which are equivalent to the original constraint:
		\begin{small}\[
		\begin{split}
		\lambda_{0,0} + \lambda_{2,0} \cdot c_{2,0,0} &= c_{1,1,0} + c_{1,1,1} + c_{1,1,2}\\
		\lambda_{1,0} +	\lambda_{2,0} \cdot c_{2,0,1} &= 0\\
		\lambda_{0,1} + \lambda_{2,0}\cdot c_{2,0,2} + \lambda_{2,1} \cdot c_{2,0,0} &= c_{1,1,1} + 2c_{1,1,2} \\
		\lambda_{1,1} + \lambda_{2,1} \cdot c_{2,0,1} &= 0\\
		\lambda_{2,1} \cdot c_{2,0,2} &= c_{1,1,2}
		\end{split}
	\]	\end{small}

\paragraph{Handling Initial Conditions} To ensure existence of $s_\init \in \initstates$ such that the generated EBRF has non-negative value on $s_\init$ (as required by the definition of an EBRF), the algorithm adds $[t_{x_0} \geq 0 \wedge f_{(\init,0)}(t_{x_0}) \geq 0]$ to $\Gamma$.
	
\paragraph{Solving the System} The EBRF shown in the Example of Section~\ref{sec:witness} satisfies the above constraints and corresponds to the following solution to the QP instance $\Gamma:$
	\[
	\begin{small}
	\begin{matrix}
		c_{0,0,0}=3 &~ c_{0,0,1}=1 &~ c_{0,0,2}=0 &~
		c_{1,0,0}=2 &~ c_{1,0,1}=1 &~ c_{1,0,2}=0\\
		c_{2,0,0}=1 &~ c_{2,0,1}=1 &~ c_{2,0,2}=0 &~
		c_{1,1,0}=0 &~ c_{1,1,1}=0 &~ c_{1,1,2}=0\\
		c_{3,0,0}=-1 &~ c_{3,0,1}=0 &~ c_{3,0,2}=0 &~~~
		c_{t,0,0}=-1 &~ c_{t,0,1}=0 &~ c_{t,0,2}=0
		 & t_{x_0} = 0
	\end{matrix}
\end{small}
	\]
	This solution was obtained by passing $\Gamma$ to the Z3 SMT solver.

\section{Experimental results on Non-Linear Benchmarks}
\label{app:non-lin}
\renewcommand{\checkmark}{\textcolor{teal}{\textbf{\ding{52}}}}
\renewcommand{\xmark}{\textcolor{teal}{\textbf{\ding{55}}}}
\newcommand{\failed}{\textcolor{red}{\textbf{F}}}
\begin{table}
	\centering
	\resizebox{!}{6.0cm}{
	\begin{tabular}{|c|c||c|c|c|c||c|c|c|c||c|c|c|c|}
		\hline 
		 & \multirow{2}{*}{Benchmark} & \multicolumn{4}{c||}{Ours} & \multicolumn{4}{c||}{nuXmv} & \multicolumn{4}{c|}{MuVal}\\
		\cline{3-14}
		& & $RA$ & $OV$ & $RC$ & $PR$ & $RA$ & $OV$ & $RC$ & $PR$ & $RA$ & $OV$ & $RC$ & $PR$\\
		\hline 
		\multirow{10}{*}{\rotatebox[origin=c]{90}{TermComp Benchmarks}} & \makecell{\texttt{...2008-aaron12}\\\texttt{true-termination.c}} & \failed & \checkmark & \xmark & \xmark & \failed & \failed & \xmark & \xmark & \xmark & \failed & \xmark & \xmark \\
		\cline{2-14}
		& \texttt{ComplInterv.c} & \xmark & \checkmark & \xmark & \xmark & \failed & \failed & \xmark & \failed & \failed & \checkmark & \xmark & \failed\\
		\cline{2-14}
		& \texttt{DoubleNeg.c} & \checkmark & \checkmark & \xmark & \xmark & \failed & \failed & \xmark & \xmark & \checkmark & \checkmark & \xmark & \xmark\\
		\cline{2-14}
		& \texttt{Factorial.c} & \xmark & \failed & \checkmark & \checkmark & \failed & \checkmark & \failed & \failed & \failed & \checkmark & \checkmark & \checkmark\\
		\cline{2-14}
		& \texttt{LogMult.c} & \checkmark & \checkmark & \checkmark & \checkmark & \failed & \checkmark & \checkmark & \checkmark & \checkmark & \checkmark & \checkmark & \checkmark\\
		\cline{2-14}
		& \texttt{svcomp_ex1.c} & \checkmark & \checkmark & \failed & \failed & \failed & \checkmark & \xmark & \xmark & \checkmark & \checkmark & \xmark & \xmark\\
		\cline{2-14}
		& \texttt{svcomp_ex2.c} & \failed & \checkmark & \checkmark & \checkmark & \failed & \checkmark & \checkmark & \checkmark & \checkmark & \checkmark & \checkmark & \failed\\
		\cline{2-14}
		& \texttt{svcomp_ex3a.c} & \checkmark & \checkmark & \checkmark & \xmark & \failed & \checkmark & \checkmark & \xmark & \checkmark & \checkmark & \checkmark & \xmark\\
		\cline{2-14}
		& \texttt{svcomp_ex3b.c} & \checkmark & \checkmark & \checkmark & \xmark & \failed & \checkmark & \checkmark & \xmark & \checkmark & \checkmark & \checkmark & \xmark\\
		\cline{2-14}
		& \texttt{svcomp_fermat.c} & \checkmark & \checkmark & \checkmark & \checkmark & \failed & \failed & \failed & \failed & \failed & \checkmark & \checkmark & \checkmark\\
		\hline 
		\hline
		\multirow{21}{*}{\rotatebox[origin=c]{90}{SV-Comp Benchmarks}} & \texttt{bresenham-ll.c} & \checkmark & \checkmark & \xmark & \xmark & \failed & \failed & \failed & \failed & \failed & \checkmark & \xmark & \xmark\\
		\cline{2-14}
		& \texttt{cohencu-ll.c} & \checkmark & \checkmark & \xmark & \xmark & \failed & \failed & \xmark & \xmark & \failed & \checkmark & \failed & \failed\\
		\cline{2-14}
		& \texttt{cohendiv-ll.c} & \checkmark & \checkmark & \checkmark & \xmark & \failed & \failed & \checkmark & \xmark & \checkmark & \checkmark & \checkmark & \xmark\\
		\cline{2-14}
		& \texttt{dijkstra-u.c} & \checkmark & \checkmark & \checkmark & \failed & \failed & \failed & \checkmark & \xmark & \failed & \checkmark & \failed & \xmark\\
		\cline{2-14}
		& \texttt{divbin.c} & \failed & \checkmark & \failed & \failed & \failed & \failed & \failed & \failed & \failed & \failed & \failed & \failed\\
		\cline{2-14}
		& \texttt{egcd2-ll.c} & \checkmark & \failed & \failed & \checkmark & \failed & \failed & \xmark & \checkmark & \failed & \failed & \failed & \failed\\
		\cline{2-14}
		& \texttt{egcd3-ll.c} & \checkmark & \failed & \failed & \checkmark & \failed & \failed & \xmark & \checkmark & \failed & \failed & \failed & \failed\\
		\cline{2-14}
		& \texttt{egcd-ll.c} & \xmark & \checkmark & \xmark & \xmark & \failed & \failed & \xmark & \xmark & \failed & \checkmark & \failed & \failed\\
		\cline{2-14}
		& \texttt{geo1-ll.c} & \checkmark & \checkmark & \checkmark & \checkmark & \failed & \failed & \checkmark & \checkmark & \checkmark & \checkmark & \checkmark & \checkmark\\
		\cline{2-14}
		& \texttt{geo2-ll.c} & \checkmark & \checkmark & \checkmark & \checkmark & \failed & \failed & \checkmark & \checkmark & \checkmark & \checkmark & \checkmark & \checkmark\\
		\cline{2-14}
		& \texttt{geo3-ll.c} & \checkmark & \checkmark & \xmark & \xmark & \failed & \failed & \xmark & \xmark & \failed & \failed & \xmark & \xmark\\
		\cline{2-14}
		& \texttt{hard-ll.c} & \checkmark & \failed & \failed & \failed & \failed & \failed & \failed & \failed & \checkmark & \checkmark & \xmark & \xmark\\
		\cline{2-14}
		& \texttt{lcm1.c} & \failed & \failed & \checkmark & \checkmark & \failed & \failed & \checkmark & \checkmark & \failed & \failed & \failed & \failed\\
		\cline{2-14}
		& \texttt{lcm2.c} & \checkmark & \checkmark & \checkmark & \checkmark & \failed & \failed & \checkmark & \checkmark & \checkmark & \checkmark & \checkmark & \failed\\
		\cline{2-14}
		& \texttt{mannadiv.c} & \checkmark & \checkmark & \checkmark & \checkmark & \failed & \failed & \checkmark & \checkmark & \checkmark & \checkmark & \checkmark & \failed\\
		\cline{2-14}
		& \texttt{ps2-ll.c} & \checkmark & \checkmark & \checkmark & \xmark & \failed & \failed & \checkmark & \xmark & \checkmark & \checkmark & \checkmark & \xmark\\
		\cline{2-14}
		& \texttt{ps3-ll.c} & \checkmark & \checkmark & \checkmark & \xmark & \failed & \failed & \checkmark & \xmark & \checkmark & \checkmark & \checkmark & \xmark\\
		\cline{2-14}
		& \texttt{ps4-ll.c} & \checkmark & \checkmark & \checkmark & \xmark & \failed & \failed & \checkmark & \xmark & \checkmark & \checkmark & \checkmark & \xmark\\
		\cline{2-14}
		& \texttt{ps5-ll.c} & \checkmark & \checkmark & \checkmark & \xmark & \failed & \failed & \checkmark & \xmark & \checkmark & \checkmark & \checkmark & \xmark\\
		\cline{2-14}
		& \texttt{ps6-ll.c} & \checkmark & \checkmark & \checkmark & \xmark & \failed & \failed & \checkmark & \xmark & \checkmark & \checkmark & \checkmark & \xmark\\
		\cline{2-14}
		& \texttt{sqrt1-ll.c} & \checkmark & \checkmark & \checkmark & \xmark & \checkmark & \checkmark & \checkmark & \xmark & \checkmark & \checkmark & \checkmark & \xmark\\
		\hline 
		\hline 
		& Total Successful Instances & 27 & 26 & 26 & 27 & 1 & 7 & 26 & 25 & 19 & 25 & 24 & 21 \\ \hline \hline
		& Unique Instances & 8 & 2 & 0 & 1 & 0 & 0 & 2 & 0 & 1 & 0 & 1 & 1 \\ \hline \hline
		& Average Time (s) & 26.0 & 29.2 & 20.2 & 53.4 & 16.6 & 164.8 & 0.2 & 0.3  & 216.5 & 43.9 & 71.3 & 26.5 \\
		\hline 
	\end{tabular}}
	\caption{Experimental results on polynomial benchmarks. $\checkmark$ denotes successful proof of the property and $\xmark$ denotes successful refutation. $\failed$ means the tool failed to decide the LTL specification.}
	\label{experiments_results_poly}
\end{table}




\end{document}